%% file: ex_article.tex
\documentclass[hidelinks,onefignum,onetabnum]{siamart251216}

\input{ex_shared}\usepackage{tabularx} \usepackage{subcaption}

\ifpdf
\hypersetup{
  pdftitle={A quantum optimization framework for data--assimilation--augmented parameter estimation},
  pdfauthor={M. J. Ahmad, M. Mohammadisiahroudi, A. Biswas and K. Hoffman}
}
\fi

\newcommand{\commentout}[1]{}
\begin{document}

\maketitle

\begin{abstract}
Parameter estimation is a fundamental challenge in the calibration of ordinary differential equation (ODE) models, where repeated numerical integration can lead to high computational cost. In this work, we investigate whether quantum algorithms can be leveraged to assist parameter estimation in nonlinear dynamical systems. We develop a hybrid classical--quantum framework that reformulates a data-assimilation-augmented parameter estimation problem as a combinatorial optimization task. Model dynamics and data assimilation are enforced entirely on the classical side, while the resulting parameter estimation cost functional is discretized and approximated by a quadratic unconstrained binary optimization (QUBO) surrogate. This surrogate is mapped to an Ising Hamiltonian, and quantum optimizers are used to search for low-energy configurations corresponding to candidate parameter estimates. We apply the framework to SIS and SIR epidemic models, the chaotic Lorenz--63 system, and a high-dimensional two-layer Lorenz--96 system. In this setting, the method is used to recover classical system parameters from partial state observations across steady-state, chaotic, and high-dimensional multiscale dynamical systems. Numerical experiments with synthetic data show that the proposed approach accurately recovers parameters while requiring data-assimilation solves only on a prescribed coarse grid. The framework avoids quantum state tomography, illustrating a viable pathway for integrating quantum optimization into data-driven parameter estimation for nonlinear dynamical systems.
\end{abstract}



\section{Introduction}
\label{Intro}

Ordinary differential equation (ODE) models play a central role in the mathematical modeling of time-dependent processes across the natural and applied sciences. In epidemiology, compartmental models such as the susceptible--infected--susceptible (SIS) and susceptible--infected--recovered (SIR) systems provide interpretable descriptions of disease transmission, recovery, and population-level progression through epidemiological states \cite{kermack,hethcote2000}. In atmospheric science and nonlinear dynamics, the Lorenz--63 and Lorenz--96 systems are standard benchmark models for studying predictability, chaos, and multiscale dynamical behavior \cite{lorenz1963,lorenz1996,sparrow1982}. These models are widely used because their parameters often have direct physical, biological, or epidemiological interpretations, and because their trajectories can be compared with observational data.

The practical use of such models is critically dependent on the accurate estimation of unknown parameters. Even when the governing equations are known, uncertainty in parameters can lead to substantial errors in predicted dynamics. In epidemiological models, transmission and recovery rates determine the timing, magnitude, and duration of outbreaks, and inaccurate parameter values may lead to unreliable forecasts or misleading assessments of intervention strategies \cite{Li2005Parameter,roda}. In chaotic systems, the challenge is even more pronounced: small errors in parameters or initial conditions can grow rapidly over time, producing trajectories that diverge from the observed dynamics even when the underlying model is deterministic \cite{lorenz1963,sparrow1982}. Parameter estimation is therefore a fundamental component of forward prediction and uncertainty quantification.

Parameter estimation in dynamical systems is naturally formulated as an inverse problem. Given a collection of observations, one seeks parameter values that produce model trajectories consistent with the available data. A common approach is to pose this task as an optimization problem constrained by governing differential equations. In particular, the unknown parameters are determined by minimizing a data-misfit cost functional subject to the ODE dynamics. This leads to an ODE-constrained optimization problem in which each evaluation of the cost functional requires the numerical integration of the underlying dynamical system. Since the observable-to-parameter map is often nonlinear and nonconvex, classical solution methods may require many repeated forward simulations before convergence. The resulting computational burden can become significant when the parameter space is high-dimensional, when observations are partial or noisy, or when the underlying dynamics exhibit sensitive dependence on parameters and initial conditions. Consequently, the development of efficient computational strategies for parameter estimation remains an important challenge in scientific computing and inverse problems \cite{bankskunisch,ascher2011,tarantola2005}.

Recent advances in quantum computing have motivated the exploration of quantum algorithms as alternative computational tools for challenging optimization tasks. In particular, hybrid quantum--classical methods seek to exploit the strengths of both paradigms by combining conventional numerical computation with quantum optimization subroutines. Within this context, the present work investigates whether quantum optimization techniques can assist parameter estimation in nonlinear dynamical systems. The goal is not to replace classical time integration, but to reformulate the parameter-search component in a form that can be treated by quantum-compatible optimization methods. This provides a pathway for incorporating quantum optimization into inverse problems governed by differential equations while remaining compatible with quantum devices \cite{farhi2014qaoa,preskill2018nisq,lucas2014ising}.

\subsection{Related Work}

\paragraph{\textbf{Conventional Parameter Estimation}} 
Conventional parameter estimation methods include nonlinear least-squares \cite{Li2005Parameter}, Bayesian \cite{BoerschSupan2017, gmd-14-4319-2021, Ghasemi2011}, data-assimilation \cite{NEWEY2025114121, Carlson2022, Ahmad2025DA, Martinez2024}, and machine-learning-based approaches \cite{YANG2025112649, Ahmad2024, PSOparamestimate, alavanifractional, param1}. These methods have been applied successfully to many dynamical systems, but their performance can be affected by limited or partial observations and by the computational expense of repeated model integration.

Data assimilation provides a natural framework for incorporating observational data into dynamical models \cite{LawStuartZygalakis2015,Kalnay2003,Asch2016,carrassi2018}. These methods combine model dynamics with available observations to improve state reconstruction and parameter estimation, including in systems with partial or uncertain data \cite{Carlson2022,Martinez2024,NEWEY2025114121,Azouani2014,Ahmad2025DA}. In parameter estimation, data assimilation can be used to construct an observation-informed cost functional that measures the agreement between model-generated trajectories and the available data. This provides a stable basis for estimating unknown parameters, particularly when the initial state is uncertain or only a subset of the system variables is observed. A useful feature of this formulation is that the nudged system need not be initialized at the true initial condition; the feedback term can stabilize the observed components and produce a meaningful cost functional landscape even when the initial state is uncertain \cite{Ahmad2025DA}. 

\paragraph{\textbf{Quantum Differential Equation Solvers}}

Quantum algorithms for differential equations have attracted increasing attention because many scientific and engineering applications are governed by dynamical systems. Several approaches have been proposed for solving ordinary and partial differential equations on quantum computers. One major direction relies on quantum linear algebra, where differential equations are discretized and reformulated as linear systems that can be solved using quantum linear system algorithms (QLSAs) \cite{berry2017quantum,childs2021highprecision}. Another line of work studies nonlinear differential equations through linear representations and Carleman-type embeddings, enabling quantum algorithms for selected classes of nonlinear systems \cite{liu2021efficient,jin2023linearrepresentation}. More recently, Schrödingerization techniques have been introduced to transform differential equations into Hamiltonian evolution problems that are amenable to quantum simulation \cite{jin2024schrodingerisation}. Variational and time-marching approaches have also been explored, where hybrid quantum--classical algorithms propagate the solution through a sequence of optimization steps using parameterized quantum circuits \cite{kyriienko2021solving}.

Despite these advances, directly solving nonlinear ODEs on quantum hardware remains challenging. Existing methods often require restrictive assumptions on sparsity, block encodings, state preparation, measurement access, smoothness of the dynamics, or the availability of sufficiently deep quantum circuits \cite{liu2021efficient,jin2023linearrepresentation}. Furthermore, parameter estimation problems typically require repeated evaluations of the underlying dynamical system for many candidate parameter values, which can significantly amplify the computational cost of quantum simulation. These limitations make direct quantum time integration difficult for the nonlinear and partially observed systems considered in this work. Consequently, rather than using quantum resources to solve the differential equations themselves, we employ classical simulation and data assimilation to construct the parameter-estimation cost functional and use quantum optimization only for the resulting combinatorial search problem. This strategy avoids the challenges associated with quantum simulation of nonlinear dynamics while remaining compatible with quantum optimization methods.

\paragraph{\textbf{Quantum Optimization}}
Recent developments in quantum computing have motivated the study of quantum and hybrid quantum--classical algorithms for optimization, inverse problems, and scientific computing. Quantum systems provide computational mechanisms that differ fundamentally from classical digital computation. Superposition allows an \(n\)-qubit register to represent amplitudes over \(2^n\) computational basis states, entanglement provides nonclassical correlations among qubits, and tunneling effects have motivated quantum optimization strategies for exploring complicated energy landscapes \cite{NielsenChuang,KadowakiNishimori1998,AlbashLidar2018}. These features do not imply an automatic speedup for every problem, but they provide a foundation for designing quantum algorithms that may offer advantages for selected structured computational tasks.

Among quantum algorithms, the Quantum Approximate Optimization Algorithm (QAOA) has received substantial attention as a hybrid variational method for combinatorial optimization \cite{FarhiQAOA}. QAOA approximately minimizes a cost functional encoded as an Ising Hamiltonian by alternating between a problem Hamiltonian and a mixing Hamiltonian, with variational parameters optimized by a classical outer loop. Since many combinatorial optimization problems can be written as quadratic unconstrained binary optimization (QUBO) problems and mapped to equivalent Ising Hamiltonians, QUBO provides a natural bridge between classical binary optimization and QAOA-based quantum optimization \cite{Ruan2023QuantumOpt}. This makes QAOA relevant for parameter estimation only after the original continuous inverse problem has been reformulated into a finite-dimensional binary optimization problem.

Quantum annealing provides another optimization paradigm for QUBO and Ising-type problems. In quantum annealing, the optimization problem is encoded into an energy landscape, and the algorithm seeks low-energy configurations by evolving the system from an initial Hamiltonian to a problem Hamiltonian. This approach is especially natural for QUBO formulations because binary variables can be mapped directly to Ising spins, and low-energy spin configurations correspond to candidate minimizers of the original binary cost functional \cite{KadowakiNishimori1998,AlbashLidar2018}. In this work, we include simulated quantum annealing (SQA) as an additional QUBO solver, allowing comparison with the QAOA-based IBM simulator and IBM Kingston QPU implementations.

Beyond QAOA, several other quantum optimization paradigms have been proposed for continuous and structured optimization problems arising in scientific computing. More recently, Decoded Quantum Interferometry (DQI) has been introduced as a quantum optimization framework that exploits the Fourier structure of objective functions to reduce certain optimization problems to decoding problems, providing superpolynomial speedups over the best known classical algorithms for specific structured optimization problems \cite{jordan2025optimization}.

Alternatively, several other quantum optimization paradigms have been proposed for continuous and structured optimization problems arising in scientific computing. Quantum Gibbs-sampling methods prepare thermal states of suitably designed Hamiltonians and have been studied for optimization, machine learning, and probabilistic inference, where low-energy states correspond to high-quality solutions \cite{Brandao2017_quantum, Apeldoorn2018_Improvements}. In addition, Quantum Hamiltonian Descent (QHD) has recently been proposed as a quantum counterpart of classical gradient-based optimization \cite{Leng2023QHD}. Derived from the path-integral formulation of dynamical systems associated with the continuous-time limit of gradient descent, QHD describes optimization as a quantum Hamiltonian evolution and exploits quantum-mechanical effects such as tunneling to explore nonconvex landscapes. The resulting dynamics can be implemented on both digital and analog quantum platforms, and empirical studies have demonstrated promising performance on nonconvex optimization problems \cite{Leng2023QHD}. More recently, theoretical and computational developments have extended QHD to non-smooth optimization problems and software frameworks for practical deployment on quantum hardware \cite{Leng2025NonSmoothQHD,Kushnir2024QHDOPT}.

For constrained optimization, quantum interior point methods (QIPMs) combine classical interior point frameworks with quantum linear algebra subroutines to accelerate the solution of the Newton systems that arise during each iteration \cite{kerenidisParkas2020_quantum, augustino2021quantum, mohammadisiahroudi2024efficient}. These methods have been developed for linear \cite{mohammadisiahroudi2025quantum}, quadratic \cite{wu2023inexact}, semidefinite \cite{mohammadisiahroudi2025quantum}, and conic optimization problems \cite{Kerenidis2019_Quantum} and provide some of the strongest theoretical complexity guarantees currently known for quantum optimization in continuous domains. Unlike QAOA and other variational methods, QIPMs are primarily designed for fault-tolerant quantum computers and rely heavily on quantum linear system algorithms and efficient state preparation. Together, QAOA, Gibbs-sampling approaches, QHD, and QIPMs illustrate the diversity of current quantum optimization research, ranging from variational algorithms to asymptotically efficient quantum algorithms for large-scale optimization.

\paragraph{\textbf{Motivation for Quantum Optimization}}

The preceding discussion suggests that the main computational bottleneck is not the lack of a differential equation solver, but the repeated search over candidate parameter values. Classical optimization methods can require many forward model evaluations, while direct quantum time integration is not yet practical for the nonlinear ODE systems considered here. This motivates a hybrid formulation in which data assimilation is used to construct a stabilized parameter-estimation cost functional, and quantum optimization is applied only after this cost functional has been converted into a finite-dimensional discrete search problem. 

Retaining the data-assimilation component is important because full-state
observations are rarely available in practical parameter estimation problems.
Full-state observation represents the ideal setting for trajectory-based
estimation, since every component of the model state can be compared directly
with the model output. In many applications, however, only a subset of the
state variables is observed. Data assimilation uses these partial observations,
together with the model dynamics, to stabilize the observed components and
recover information about the unobserved components. Therefore, we retain the
data-assimilation formulation in the proposed quantum optimization framework
so that the resulting parameter estimator remains applicable to partially
observed dynamical systems.

QUBO is a natural formulation for this discrete optimization stage. The cost functional is built from squared residuals and therefore has an intrinsic least-squares structure. Although the dependence of these residuals on the unknown parameters is generally nonlinear, the region of the parameter domain with small cost functional values can be approximated by a quadratic surrogate fitted from coarse-grid data-assimilation evaluations. Once the refined-grid parameter search is encoded using binary variables, this quadratic surrogate leads directly to a QUBO problem. Moreover, QUBO cost functionals admit a standard mapping to Ising Hamiltonians, making them compatible with QAOA and other quantum optimization methods. Thus, the QUBO formulation is not introduced as an arbitrary discretization; it is the mechanism that connects the classical cost functional to a quantum-compatible optimization problem. This also separates the role of the quantum algorithm from the dynamical simulation: the quantum routine is not used to solve the ODE, but only to search the binary-encoded surrogate landscape obtained from classical data-assimilation evaluations.

\subsection{Contributions}

The main contribution of this paper is a hybrid data-assimilation and quantum-optimization framework for parameter estimation in nonlinear ODE systems. The method uses coarse-grid cost functional values to construct a quadratic surrogate, encodes the refined parameter search as a QUBO, maps the resulting binary cost functional to an Ising Hamiltonian, and applies quantum optimizers to obtain parameter estimates.

The main contributions of this work are as follows:
\begin{itemize}
    \item We develop a hybrid classical--quantum framework for parameter estimation in nonlinear ODE systems by combining nudging-based data assimilation with QUBO-based quantum optimization.

    \item We introduce a coarse-to-refined surrogate strategy in which expensive cost functional values are computed only on a coarse grid, while a QUBO-based quantum optimization stage searches over a refined binary-encoded parameter grid.

    \item We formulate the parameter estimation problem so that all ODE integrations remain classical, avoiding the need for quantum simulation of the underlying differential equations.

    \item We obtain parameter estimates directly from measured bitstrings, so the proposed approach does not require quantum state tomography \cite{Kaznady2008QuantumStateTomography}.

    \item We provide theoretical statements and fitting-error diagnostics that clarify how surrogate approximation, QUBO fitting, and refined-grid resolution affect the recovered parameter estimate.

    \item We test the method on four model problems: the SIS and SIR epidemic models, whose trajectories approach stable equilibria under the regimes considered; the Lorenz--63 system in a chaotic regime; and a high-dimensional, multiscale two-layer Lorenz--96 system.
\end{itemize}

The remainder of the paper is organized as follows. In Section~\ref{PE}, we introduce the general parameter estimation formulation and describe the data-assimilation-augmented cost functional construction that serves as the starting point for the proposed method. In Section~\ref{PEQC}, we present the proposed quantum parameter estimation framework, including coarse-grid cost functional evaluation, continuous quadratic surrogate fitting, refined-grid binary encoding, QUBO construction, Ising mapping, and quantum optimization. We also establish theoretical results describing how surrogate approximation error and refined-grid resolution affect the recovered parameter estimate. In Section~\ref{models}, we present numerical experiments for the SIS, SIR, Lorenz--63, and two-layer Lorenz--96 systems using QAOA on an IBM simulator, QAOA on the IBM Kingston QPU, and SQA.  Through these examples, we assess the method across epidemiological, chaotic, and high-dimensional multiscale dynamical systems, with particular attention to partial observations and refined-grid parameter recovery. Finally, in Section~\ref{D}, we summarize the main findings and outline directions for future work.

\section{Preliminaries}
\label{PE}

This section introduces the general dynamical-system and observation setting
used throughout the paper. It then presents the
data-assimilation-augmented parameter estimation formulation developed in
\cite{Ahmad2025DA}, which serves as the starting point for the proposed quantum
optimization framework.

Consider a nonlinear dynamical system governed by the ordinary differential equation
\begin{equation}
\dot{x}(t)=F(x(t),\theta), 
\qquad 
x(0)=x_0,
\label{eq:general_ode}
\end{equation}
where \(x(t)\in\mathbb{R}^d\) denotes the state vector,
\(\theta\in\Theta\subset\mathbb{R}^{d_\theta}\) is the unknown parameter vector,
and \(F:\mathbb{R}^d\times\Theta\to\mathbb{R}^d\) is assumed to be sufficiently smooth.
Let the observations be defined through the linear observation operator
\begin{equation}
y(t)=Hx(t),
\label{eq:obs_op}
\end{equation}
where \(y(t)\in\mathbb{R}^m\), \(m\) is the number of observed quantities, and
\(H\in\mathbb{R}^{m\times d}\) maps the full state to the observation space.
In general, \(H\) may select individual state components or form linear combinations
of the state variables. In the numerical examples considered in this work,
\(H\) is a coordinate projection that selects the observed state components.
Observations are available at times \(\{t_k\}_{k=1}^{N_t}\) and are denoted by
\(y^{\mathrm{data}}(t_k)\). The goal of parameter estimation is to determine
parameter values for which the model output agrees with the observed data.

\paragraph{\textbf{Data--Assimilation--Augmented Parameter Estimation}}
 
\label{PECC}

The parameter estimation approach developed in
\cite{Ahmad2025DA} uses a nudged dynamical system to stabilize the estimation
process by incorporating observational data through a feedback term. For a candidate parameter vector \(\tilde{\theta}\), the nudged system is defined by
\begin{equation}
\dot{\tilde{x}}(t)
=
F(\tilde{x}(t),\tilde{\theta})
+
\mu\,H^T\!\left(
y^{\mathrm{data}}(t)-H(\tilde{x}(t))
\right),
\label{eq:nudged_system}
\end{equation}
where \(\mu>0\) is the nudging coefficient. The initial condition for the nudged system need not coincide with the true initial condition. Instead, the nudging term drives the observed components of the nudged trajectory toward the available data.

Using the nudged trajectory, we define the data-misfit cost functional
\begin{equation}
C_{\mathrm{nudge}}(\tilde{\theta};\mu)
=
\sum_{k=1}^{N_t}
\left\|
y^{\mathrm{data}}(t_k)
-
H\!\bigl(\tilde{x}(t_k;\tilde{\theta},\mu)\bigr)
\right\|^2,
\label{eq:cost_nudge}
\end{equation}
where \(N_t\) is the number of observation times. The data-assimilation-augmented parameter estimation problem is then formulated as
\begin{equation}
\tilde{\theta}^{\ast}
\in
\arg\min_{\tilde{\theta}\in\Theta_{\mathrm{s}}}
C_{\mathrm{nudge}}(\tilde{\theta};\mu),
\label{eq:classical_min}
\end{equation}
where \(\Theta_{\mathrm{s}}\subset\Theta\) is a bounded parameter search domain over which the minimization is performed. The resulting data-assimilation-augmented parameter estimation procedure is summarized in Algorithm~\ref{algo1}.

\begin{algorithm}[h]
\caption{Data--Assimilation--Augmented Parameter Estimation}
\label{algo1}
\begin{algorithmic}[1]
\Require Observation times \(\{t_k\}_{k=1}^{N_t}\), data \(\{y^{\mathrm{data}}(t_k)\}_{k=1}^{N_t}\), bounded parameter search domain \(\Theta_{\mathrm{s}}\subset\Theta\), nudging coefficient \(\mu>0\), initial guess \(\tilde{\theta}^{(0)}\), stopping tolerance \(\varepsilon\)
\State Define the nudged ODE system \eqref{eq:nudged_system} with parameter vector \(\tilde{\theta}\) and nudging coefficient \(\mu\)
\State Define the cost functional \(C_{\mathrm{nudge}}(\tilde{\theta};\mu)\) in \eqref{eq:cost_nudge}
\State Initialize the optimizer at \(\tilde{\theta}^{(0)}\)
\While{optimizer has not converged}
    \State Propose a candidate parameter vector \(\tilde{\theta}\)
    \State Solve the nudged system \eqref{eq:nudged_system} to obtain \(\tilde{x}(t_k;\tilde{\theta},\mu)\)
    \State Evaluate \(C_{\mathrm{nudge}}(\tilde{\theta};\mu)\)
    \State Update the optimizer state
\EndWhile 
\State \textbf{return} \(\tilde{\theta}^{\ast}\)
\end{algorithmic}
\end{algorithm}

In practice, the minimization in \eqref{eq:classical_min} may be performed using a numerical optimization method such as the Nelder--Mead simplex method \cite{nelder1965simplex}. Each evaluation of the cost functional requires solving the nudged ODE system forward in time, so repeated numerical integration is the dominant computational cost. For further details on this data-assimilation-augmented parameter estimation method, we refer the reader to \cite{Ahmad2025DA}.

\section{Proposed Quantum Optimization Framework}
\label{PEQC}

We propose a quantum-enabled extension of the data-assimilation-augmented optimization framework by reformulating the minimization of \(C_{\mathrm{nudge}}(\tilde{\theta};\mu)\) as a combinatorial optimization problem. In contrast to a direct grid-based formulation, we distinguish between a coarse grid used for model-based cost functional evaluations and a refined grid used for the QUBO-based quantum search.

Let
\begin{equation}
\Theta_{\mathrm{s}}
=
\prod_{j=1}^{d_\theta}
[\theta_j^{\min},\theta_j^{\max}]
\label{eq:Theta_general}
\end{equation}
denote the bounded parameter search domain. Let \(M_c\) and \(M_f\) denote the number of coarse and refined grid points per parameter, respectively, with \(M_f\ge M_c\). The coarse grid for the \(j\)-th parameter is given by
\begin{equation}
\tilde{\theta}_{j,c}^{(i)}
=
\theta_j^{\min}
+
\frac{i}{M_c-1}
\bigl(\theta_j^{\max}-\theta_j^{\min}\bigr),
\qquad
i=0,\ldots,M_c-1.
\label{eq:coarse_grid_param}
\end{equation}
For \(d_\theta\) unknown parameters, this gives \(L_c=M_c^{d_\theta}\) total coarse--grid points. Each coarse grid point defines a candidate parameter vector \(\tilde{\theta}_{c,\ell}\), \(\ell=1,\ldots,L_c\). For each candidate, we solve the nudged system \eqref{eq:nudged_system} and compute $C_\ell
=
C_{\mathrm{nudge}}(\tilde{\theta}_{c,\ell};\mu).$
Thus, the expensive forward integrations and cost functional evaluations are performed only on the coarse grid.

The resulting coarse-grid cost functional values are then used to fit a continuous quadratic surrogate over the parameter domain. Let $\theta
=
(\theta_1,\ldots,\theta_{d_\theta})\in\Theta_{\mathrm{s}}$
denote a generic parameter vector, where \(\theta_j\) is the \(j\)-th unknown parameter. For example, in the SIR model, \(\theta=(\beta,\gamma)\), while in the Lorenz--63 model, \(\theta=(\sigma,\rho,\beta)\). We approximate the cost functional by a quadratic function of the parameter components:
\begin{equation}
\widehat{C}(\theta)
=
c_0
+
\sum_{j=1}^{d_\theta} c_j \theta_j
+
\sum_{1\le i\le j\le d_\theta} c_{ij}\theta_i\theta_j .
\label{eq:continuous_quadratic_surrogate}
\end{equation}
Here, \(c_0\) is the constant coefficient, \(c_j\) are the linear coefficients, and \(c_{ij}\) are the quadratic interaction coefficients. The terms with \(i=j\) represent squared parameter contributions, \(\theta_i^2\), while the terms with \(i<j\) represent pairwise interactions between different parameters. The coefficients
\(
\{c_0,\ c_j,\ c_{ij}\}
\)
are determined by weighted least-squares fitting using the coarse-grid data $\left\{
\left(\tilde{\theta}_{c,\ell},C_\ell\right)
\right\}_{\ell=1}^{L_c}.$
The use of least-squares fitting is appropriate in this setting for two related reasons. First, the cost functional is itself defined as a sum of squared residuals, so a least-squares surrogate is consistent with the structure of the underlying misfit functional. Second, the quadratic surrogate is linear in the unknown coefficients \(\{c_0,c_j,c_{ij}\}\), which means that the surrogate coefficients can be computed by solving a linear least-squares problem rather than a nonlinear optimization problem.

To define the weights, let $C_{\min}=\min_{1\le \ell\le L_c} C_\ell,
C_{\max}=\max_{1\le \ell\le L_c} C_\ell .$
We first normalize the coarse-grid cost functional values by
\[
\widetilde C_\ell
=
\frac{C_\ell-C_{\min}}
{C_{\max}-C_{\min}+\delta},
\]
where \(\delta>0\) is a small numerical constant used to avoid division by zero. The least-squares weights are then chosen as $w_\ell=\exp(-\lambda \widetilde C_\ell).$
In the numerical experiments, we use \(\lambda=8\). This weighting assigns the largest weights to the smallest cost functional values and progressively downweights regions with large cost functional values. Thus, the fitted quadratic surrogate is biased toward accuracy in the region with small cost functional values where the parameter minimizer is expected to lie.

Equivalently, let
\[
q
=
1+d_\theta+\frac{d_\theta(d_\theta+1)}{2}
\]
denote the total number of constant, linear, squared, and pairwise interaction
features in the quadratic surrogate. Define the quadratic feature vector by $\phi(\theta)
=
\left(
1,\,
\theta_1,\ldots,\theta_{d_\theta},\,
\{\theta_i\theta_j\}_{1\le i\le j\le d_\theta}
\right)^T
\in\mathbb{R}^q.$
Thus, \(\phi(\theta)\) contains the constant term, all linear terms, all squared
terms, and all pairwise interaction terms. Let
\(\Phi_c\in\mathbb{R}^{L_c\times q}\) denote the coarse-grid design matrix,
whose \(\ell\)-th row is $\phi(\tilde{\theta}_{c,\ell})^T.$
Let \(C^{\mathrm{sc}}\in\mathbb{R}^{L_c}\) denote the vector of normalized
coarse-grid cost functional values, with entries $C^{\mathrm{sc}}_\ell=\widetilde C_\ell.$
The coefficient vector \(c\in\mathbb{R}^q\) is computed from
\[
c
=
\arg\min_{\bar c\in\mathbb{R}^q}
\left\|
W_c^{1/2}
\left(
\Phi_c \bar c-C^{\mathrm{sc}}
\right)
\right\|_2^2,
\]
where $W_c=\operatorname{diag}(w_1,\ldots,w_{L_c}).$ This gives the fitted continuous surrogate $\widehat C(\theta)=c^T\phi(\theta).$

Next, we introduce the refined grid
\begin{equation}
\tilde{\theta}_{j,f}^{(i)}
=
\theta_j^{\min}
+
\frac{i}{M_f-1}
\bigl(\theta_j^{\max}-\theta_j^{\min}\bigr),
\qquad
i=0,\ldots,M_f-1.
\label{eq:refined_grid_param}
\end{equation}
The surrogate \(\widehat{C}\) is evaluated on this refined grid, producing a refined discrete cost functional landscape without additional nudged ODE solves. Thus, the coarse grid controls the number of expensive model integrations, while the refined grid controls the resolution of the final binary search.

Each refined grid point is encoded by a binary decision vector \(z\in\{0,1\}^n\). If $b=\left\lceil \log_2 M_f\right\rceil$
bits are used per parameter, then $n=d_\theta b .$
For the \(j\)-th parameter, let $z^{(j)}
=
\left(z_{j,0},z_{j,1},\ldots,z_{j,b-1}\right)
\in \{0,1\}^{b}$
denote the block of bits assigned to that parameter. The corresponding refined-grid index is
\begin{equation}
i_j(z)
=
\sum_{r=0}^{b-1}2^{b-1-r} z_{j,r},
\qquad
j=1,\ldots,d_\theta .
\label{eq:binary_index_general}
\end{equation}
The decoded refined-grid parameter value is therefore
\begin{equation}
\tilde{\theta}_{j,f}(z)
=
\theta_j^{\min}
+
\frac{i_j(z)}{M_f-1}
\bigl(\theta_j^{\max}-\theta_j^{\min}\bigr),
\qquad
j=1,\ldots,d_\theta .
\label{eq:binary_decode_general}
\end{equation}
Thus, $\tilde{\theta}_f(z)
=
\left(
\tilde{\theta}_{1,f}(z),
\ldots,
\tilde{\theta}_{d_\theta,f}(z)
\right)$
denotes the refined-grid parameter vector decoded from \(z\). When \(M_f\) is a power of two, each \(b\)-bit block corresponds to a valid refined-grid index. In the numerical experiments below, \(M_f=32\), so each parameter is represented by five binary variables.

The refined surrogate values are then represented by a quadratic binary model of the form
\begin{equation}
E(z)
=
a_0
+
\sum_{i=1}^{n} a_i z_i
+
\sum_{1\le i<j\le n} a_{ij} z_i z_j,
\qquad
z_i\in\{0,1\}.
\label{eq:qubo_general}
\end{equation}
Here, \(z=(z_1,\ldots,z_n)\) is the binary vector encoding a refined-grid parameter vector, \(a_0\) is the constant term, \(a_i\) are the linear QUBO coefficients, and \(a_{ij}\) are the pairwise quadratic interaction coefficients between binary variables. These coefficients are not model parameters; they are coefficients of the binary surrogate used to represent the refined surrogate landscape on the encoded parameter grid.

The QUBO coefficients
\(
\{a_0,\ a_i,\ a_{ij}\}
\)
are obtained by weighted least-squares fitting so that $E(z)
\approx
\widehat{C}\bigl(\tilde{\theta}_f(z)\bigr)$
over the refined grid. After the continuous surrogate is evaluated on the
refined grid, the refined surrogate values are shifted and normalized to
\([0,1]\). If \(\widehat C_f(z)\) denotes the normalized refined-grid
surrogate value associated with bitstring \(z\), then the QUBO
least-squares weights are chosen as $w_q(z)
=
\exp\bigl(-\lambda \widehat C_f(z)\bigr),$
again with \(\lambda=8\).

Because the continuous surrogate is quadratic in the parameters and the
binary decoding in \eqref{eq:binary_decode_general} is affine in the binary
variables, the refined surrogate admits an exact QUBO representation. The
weighted least-squares step is therefore used as a numerical procedure for
recovering the coefficients of this binary quadratic representation. As a
result, the QUBO fitting error is expected to be near machine precision, as
observed in the numerical experiments.

The continuous surrogate coefficients in
\eqref{eq:continuous_quadratic_surrogate} are fitted using the normalized
coarse-grid cost functional values, whereas the QUBO coefficients in
\eqref{eq:qubo_general} represent the normalized continuous surrogate on the
refined binary grid. This produces the QUBO problem used in the final quantum
optimization stage.

Using the standard transformation from binary variables to Ising spin variables,
\[
z_i=\frac{1-s_i}{2},
\qquad
z_i\in\{0,1\}, \qquad s_i\in\{-1,+1\},
\]
the QUBO cost functional is mapped to an Ising Hamiltonian
\begin{equation}
\hat{H}_C
=
\kappa I
+
\sum_{i=1}^{n} h_i Z_i
+
\sum_{1\le i<j\le n}J_{ij}Z_iZ_j,
\label{eq:ising_general}
\end{equation}
where \(Z_i\) denotes the Pauli-\(Z\) operator acting on qubit \(i\). Here, \(\kappa\) is a constant energy shift, \(h_i\) are the one-qubit Ising coefficients, and \(J_{ij}\) are the two-qubit coupling coefficients. These Ising coefficients are obtained algebraically from the QUBO coefficients \(\{a_0,a_i,a_{ij}\}\). The constant shift \(\kappa I\) does not affect the minimizer, but it is included in \eqref{eq:ising_general} for completeness. 

The final step is to search for a low-energy configuration of the Ising Hamiltonian
\(\hat{H}_C\) using a quantum optimizer \(\mathcal{S}\). In this work, the
optimizer is treated as a solver for the binary energy landscape defined by the
QUBO/Ising formulation. The low-energy bitstring obtained from this optimization
step is decoded through \eqref{eq:binary_decode_general} to obtain the
corresponding refined-grid parameter estimate \(\tilde{\theta}_Q\).

This formulation is independent of the particular optimization procedure used in the final QUBO/Ising search. In the numerical experiments, we use QAOA on an IBM simulator, QAOA on the IBM Kingston QPU, and SQA.
After optimization, the solvers return candidate bitstrings. The returned low-energy bitstring is decoded to obtain a refined-grid parameter vector and hence the quantum-assisted estimate \(\tilde{\theta}_Q\). The proposed hybrid procedure is summarized in Algorithm~\ref{alg:hybrid_estimation}.

\begin{algorithm}[h]
\caption{Proposed Quantum Optimization Framework for
Data-Assimilation-Augmented Parameter Estimation}
\label{alg:hybrid_estimation}
\begin{algorithmic}[1]
\Require Observation times \(\{t_k\}_{k=1}^{N_t}\), data \(\{y^{\mathrm{data}}(t_k)\}_{k=1}^{N_t}\), bounded parameter search domain \(\Theta_{\mathrm{s}}\subset\Theta\), nudging coefficient \(\mu>0\), coarse grid size \(M_c\), refined grid size \(M_f\), QUBO/Ising optimizer \(\mathcal{S}\)
\State Discretize \(\Theta_{\mathrm{s}}\) using an \(M_c\)-point coarse grid per parameter
\For{each coarse grid point \(\tilde{\theta}_{c,\ell}\)}
    \State Solve the nudged system \eqref{eq:nudged_system} with parameter vector \(\tilde{\theta}_{c,\ell}\) and gain \(\mu\)
    \State Compute \(C_\ell=C_{\mathrm{nudge}}(\tilde{\theta}_{c,\ell};\mu)\)
\EndFor
\State Normalize the coarse-grid cost functional values to obtain \(\{\widetilde C_\ell\}\)
\State Define weights \(w_\ell=\exp(-\lambda\widetilde C_\ell)\)
\State Fit the continuous quadratic surrogate \(\widehat{C}\) in \eqref{eq:continuous_quadratic_surrogate} by weighted least squares
\State Monitor the continuous surrogate fitting errors \(\mathrm{RMSE}_{\mathrm{LS}}\) and \(e_{\infty,\mathrm{LS}}\)
\State Discretize \(\Theta_{\mathrm{s}}\) using an \(M_f\)-point refined grid per parameter
\State Evaluate \(\widehat{C}\) on the refined grid
\State Encode refined grid points by bitstrings \(z\in\{0,1\}^n\), where \(n=d_\theta\lceil \log_2 M_f\rceil\)
\State Normalize the refined surrogate values on the refined grid
\State Define QUBO weights
\(
w_q(z)=\exp(-\lambda\widehat C_f(z))
\)
\State Recover the coefficients of the QUBO representation
\(E(z)\) in \eqref{eq:qubo_general} by weighted least squares
\State Monitor the numerical QUBO representation error
\(\mathrm{MSE}_{\mathrm{QUBO}}\)
\State Map \(E(z)\) to an Ising Hamiltonian \(\hat{H}_C\)
\State Use the selected QUBO/Ising optimizer \(\mathcal{S}\) to approximately minimize the energy of \(\hat{H}_C\) and obtain a bitstring \(z^\ast\)
\State Decode \(z^\ast\) to obtain the corresponding parameter estimate \(\tilde{\theta}_Q\)
\State \Return \(\tilde{\theta}_Q\)
\end{algorithmic}
\end{algorithm}

\paragraph{\textbf{Least-Squares and QUBO Fitting Errors}}
The continuous quadratic surrogate introduces an approximation error because the cost functional is generally not exactly quadratic in the unknown parameters. We measure this error on the coarse training grid using the weighted root-mean-square error
\[
\mathrm{RMSE}_{\mathrm{LS}}
=
\left(
\frac{
\sum_{\ell=1}^{L_c}
w_\ell
\left(
\widehat C(\tilde{\theta}_{c,\ell})-\widetilde C_\ell
\right)^2
}{
\sum_{\ell=1}^{L_c} w_\ell
}
\right)^{1/2}.
\]
We monitor the maximum coarse-grid fitting error $e_{\infty,\mathrm{LS}}
=
\max_{1\le \ell\le L_c}
\left|
\widehat C(\tilde{\theta}_{c,\ell})-\widetilde C_\ell
\right|.$
These quantities measure how accurately the continuous quadratic surrogate approximates the normalized cost functional values used for fitting.

After evaluating the continuous surrogate on the refined binary grid, the
QUBO coefficients are recovered through weighted least squares. Since the
continuous surrogate is quadratic in the parameters and the binary parameter
decoding is affine, an exact binary quadratic representation exists. Let
\(\mathcal{Z}\subseteq\{0,1\}^n\) denote the set of bitstrings corresponding
to the refined-grid parameter points. The numerical accuracy of the recovered
QUBO representation is measured by
\[
\mathrm{MSE}_{\mathrm{QUBO}}
=
\frac{1}{|\mathcal{Z}|}
\sum_{z\in\mathcal{Z}}
\left(E(z)-\widehat C_f(z)\right)^2,
\]
where \(|\mathcal{Z}|\) is the number of refined-grid bitstrings,
\(\widehat C_f(z)\) is the normalized value of
\(\widehat C\bigl(\tilde{\theta}_f(z)\bigr)\), and \(E(z)\) is the recovered
QUBO energy. Thus, \(\mathrm{MSE}_{\mathrm{QUBO}}\) measures numerical
coefficient-recovery error rather than a separate model-approximation error.
 
\paragraph{\textbf{Conditioning of the Least-Squares Fits}}
A potential numerical challenge in the surrogate construction is the conditioning of the least-squares design matrices. Both the continuous quadratic surrogate and the QUBO surrogate are fitted using feature matrices whose columns may be correlated. For example, the continuous surrogate includes linear, squared, and pairwise interaction terms in the parameters, while the QUBO surrogate includes binary variables and their pairwise products. On a structured grid, and especially after weighting the low-cost region more heavily, these features can exhibit multicollinearity, so the corresponding design matrix can become ill-conditioned or, in degenerate cases, rank deficient \cite{belsley1980regression,golub2013matrix,bjorck1996numerical}. Such multicollinearity can make individual fitted coefficients sensitive to perturbations and difficult to interpret. However, in predictive or optimization-oriented uses of least-squares models, multicollinearity does not necessarily imply poor performance of the fitted response surface, since its main effect is often on the stability and interpretation of individual fitted coefficients rather than on the overall fitted response \cite{obrien2007caution,vatcheva2016multicollinearity}. In the present framework, the fitted coefficients are not interpreted as physical or model parameters; they are used only to represent the surrogate cost functional landscape and identify low-cost regions. In the present experiments, the fitted surrogates appeared numerically stable, as indicated by the monitored values of \(\mathrm{RMSE}_{\mathrm{LS}}\), \(e_{\infty,\mathrm{LS}}\), and \(\mathrm{MSE}_{\mathrm{QUBO}}\). Thus, although multicollinearity is a possible challenge of the least-squares fitting step, it did not prevent stable recovery of the low-cost parameter region in the examples considered in this work.

The construction above separates expensive model evaluation from the resolution
of the binary optimization problem. The nudged system is solved only on the
coarse grid, while the QUBO/Ising optimization stage searches over the refined
grid through the learned surrogate. Since the quantum optimizer minimizes the
QUBO representation rather than the original cost functional in
\eqref{eq:cost_nudge}, it is important to understand how closely the minimizer
of the surrogate problem approximates the minimizer of the target discrete cost
functional.

Here, \(C\) denotes the target cost functional on the refined binary grid,
while \(E\) denotes the QUBO representation minimized by the quantum optimizer.
In the present quadratic construction, the dominant approximation arises from
fitting the continuous quadratic surrogate to the coarse-grid cost functional values. The subsequent QUBO
coefficient-recovery step reproduces the refined surrogate up to numerical
roundoff, as measured by \(\mathrm{MSE}_{\mathrm{QUBO}}\). Therefore, the
uniform error parameter \(\varepsilon\) in
Theorem~\ref{thm:surrogate-stability} primarily reflects the continuous
surrogate approximation error, together with any numerical error introduced
when recovering the QUBO coefficients.

\begin{theorem}
\label{thm:surrogate-stability}
Let \(\mathcal{Z}\) be a nonempty finite set. Let \(C:\mathcal{Z}\to\mathbb{R}\) denote the target refined-grid cost functional and let \(E:\mathcal{Z}\to\mathbb{R}\) denote a surrogate cost functional. Assume that
\begin{equation}
\sup_{z\in\mathcal{Z}} |E(z)-C(z)| \le \varepsilon
\label{eq:uniform-bound}
\end{equation}
for some \(\varepsilon\ge 0\). Let $z_C\in\arg\min_{z\in\mathcal{Z}} C(z), 
z_E\in\arg\min_{z\in\mathcal{Z}} E(z).$
Then
\begin{equation}
C(z_E)
\le
\min_{z\in\mathcal{Z}} C(z) + 2\varepsilon .
\label{eq:near-opt}
\end{equation}
Moreover, assume \(z_C\) is the unique minimizer of \(C\) on \(\mathcal{Z}\), and define the optimality gap
\begin{equation}
\Delta
:=
\min_{z\in\mathcal{Z}\setminus\{z_C\}}
\bigl(C(z)-C(z_C)\bigr).
\label{eq:gap}
\end{equation}
If \(\Delta>0\) and \(\varepsilon<\Delta/2\), then \(z_C\) is the unique minimizer of \(E\) on \(\mathcal{Z}\), and hence
\begin{equation}
z_E=z_C .
\label{eq:exact-recovery}
\end{equation}
\end{theorem}

\begin{proof}
From Eq.~\eqref{eq:uniform-bound}, for all \(z\in\mathcal{Z}\),
\begin{equation}
C(z)-\varepsilon
\le
E(z)
\le
C(z)+\varepsilon .
\label{eq:sandwich}
\end{equation}
Since \(z_E\) minimizes \(E\) on \(\mathcal{Z}\), \(E(z_E)\le E(z_C)\). Using Eq.~\eqref{eq:sandwich} at \(z=z_E\) and \(z=z_C\), we obtain
\[
C(z_E)-\varepsilon
\le
E(z_E)
\le
E(z_C)
\le
C(z_C)+\varepsilon.
\]
Thus, $C(z_E)\le C(z_C)+2\varepsilon.$
Since \(z_C\in\arg\min_{\mathcal{Z}} C\), we have \(C(z_C)=\min_{z\in\mathcal{Z}}C(z)\), which proves Eq.~\eqref{eq:near-opt}.

Now assume that \(z_C\) is the unique minimizer of \(C\) and that \(\Delta\) is defined by Eq.~\eqref{eq:gap}. Then, for every \(z\in\mathcal{Z}\setminus\{z_C\}\),
\begin{equation}
C(z)\ge C(z_C)+\Delta .
\label{eq:gap-ineq}
\end{equation}
Using Eq.~\eqref{eq:sandwich} and Eq.~\eqref{eq:gap-ineq}, for \(z\neq z_C\),
\[
E(z)
\ge
C(z)-\varepsilon
\ge
C(z_C)+\Delta-\varepsilon,
\qquad
E(z_C)
\le
C(z_C)+\varepsilon.
\]
Therefore, $E(z)-E(z_C)
\ge
\Delta-2\varepsilon.$
If \(\varepsilon<\Delta/2\), then \(\Delta-2\varepsilon>0\), and hence \(E(z)>E(z_C)\) for all \(z\neq z_C\). Thus, \(z_C\) is the unique minimizer of \(E\) on \(\mathcal{Z}\), and Eq.~\eqref{eq:exact-recovery} follows.
\end{proof}

The discretization of the parameter domain introduces an approximation error because the true minimizer of the continuous cost functional may not lie exactly on the refined grid. Theorem~\ref{thm:grid-discretization-short} quantifies how this error depends on the refined grid spacing. Its assumptions are local regularity conditions near an identifiable minimizer and are intended to describe the effect of grid resolution, rather than to provide a global characterization of arbitrary nonconvex cost functional landscapes.

\noindent \begin{theorem}
\label{thm:grid-discretization-short}
Let
\[
\Theta_{\mathrm{s}}
=
\prod_{j=1}^{d_\theta}
[\theta_j^{\min},\theta_j^{\max}]
\subset \Theta
\subset\mathbb{R}^{d_\theta},
\]
and let \(C:\Theta_{\mathrm{s}}\to\mathbb{R}\) be continuously differentiable with a unique minimizer
\(\theta^\ast\in\operatorname{int}(\Theta_{\mathrm{s}})\). Assume there exist constants
\(\alpha,L>0\) such that, for all \(\theta\in\Theta_{\mathrm{s}}\),
\[
C(\theta)
\ge
C(\theta^\ast)
+
\frac{\alpha}{2}
\lVert \theta-\theta^\ast\rVert^2,
\qquad
C(\theta)
\le
C(\theta^\ast)
+
\frac{L}{2}
\lVert \theta-\theta^\ast\rVert^2 .
\]
Let \(\Theta_{h_f}\subset\Theta_{\mathrm{s}}\) be the uniform refined grid induced by
\eqref{eq:refined_grid_param}, and set
\[
h_f
:=
\max_{1\le j\le d_\theta}
\frac{\theta_j^{\max}-\theta_j^{\min}}{M_f-1}.
\]
Let $\theta_{h_f}\in\arg\min_{\theta\in\Theta_{h_f}} C(\theta),$
and assume $\operatorname{dist}_\infty(\theta^\ast,\partial\Theta_{\mathrm{s}})\ge h_f/2.$
Then
\[
0
\le
C(\theta_{h_f})-C(\theta^\ast)
\le
\frac{L}{8}\,d_\theta\,h_f^2,
\qquad
\lVert\theta_{h_f}-\theta^\ast\rVert
\le
\sqrt{\frac{L}{\alpha}}\,
\frac{\sqrt{d_\theta}}{2}\,h_f .
\]
\end{theorem}

\begin{proof}
The assumption $\operatorname{dist}_\infty
\bigl(\theta^\ast,\partial\Theta_{\mathrm{s}}\bigr)
\ge \frac{h_f}{2}$
ensures that \(\theta^\ast\) lies at least half a refined-grid spacing away
from the boundary of \(\Theta_{\mathrm{s}}\). Therefore, there exists a grid
point \(\tilde{\theta}\in\Theta_{h_f}\) such that
\[
|\tilde{\theta}_j-\theta_j^\ast|
\le \frac{h_f}{2},
\qquad
j=1,\ldots,d_\theta.
\]
Hence,
\[
\lVert \tilde{\theta}-\theta^\ast\rVert^2
\le
d_\theta(h_f/2)^2
=
\frac{d_\theta h_f^2}{4}.
\]
Using the upper quadratic bound gives
\[
C(\tilde{\theta})-C(\theta^\ast)
\le
\frac{L}{2}
\frac{d_\theta h_f^2}{4}
=
\frac{L}{8}d_\theta h_f^2 .
\]
By optimality of \(\theta_{h_f}\) on \(\Theta_{h_f}\), we have $C(\theta_{h_f})\le C(\tilde{\theta}).$
Since \(\Theta_{h_f}\subset\Theta_{\mathrm{s}}\) and \(\theta^\ast\) is the minimizer of \(C\) on
\(\Theta_{\mathrm{s}}\), we also have $C(\theta_{h_f})\ge C(\theta^\ast).$
Therefore,
\[
0
\le
C(\theta_{h_f})-C(\theta^\ast)
\le
\frac{L}{8}d_\theta h_f^2 .
\]
The parameter error bound follows by combining this estimate with the lower quadratic growth condition:
\[
\frac{\alpha}{2}
\lVert \theta_{h_f}-\theta^\ast\rVert^2
\le
C(\theta_{h_f})-C(\theta^\ast)
\le
\frac{L}{8}d_\theta h_f^2 .
\]
Taking square roots gives
\[
\lVert\theta_{h_f}-\theta^\ast\rVert
\le
\sqrt{\frac{L}{\alpha}}\,
\frac{\sqrt{d_\theta}}{2}\,h_f .
\]
\end{proof}


\section{Numerical Results}
\label{models}

This section evaluates Algorithm~\ref{alg:hybrid_estimation} on four model
problems: the SIS model, the SIR model, the Lorenz--63 system, and the
two-layer Lorenz--96 system. The SIS and SIR examples represent
epidemiological models whose trajectories approach stable equilibria under the
parameter regimes considered. The Lorenz--63 example provides a chaotic test
problem, while the two-layer Lorenz--96 example provides a high-dimensional,
multiscale dissipative system. Together, these examples evaluate the proposed
framework across steady-state, chaotic, and high-dimensional dynamical systems. A sample for is available on this \hyperlink{https://github.com/Jalilahmad4/Quantum-Opt.-for-Parameter-Estimation-in-Dynamical-Systems.git}{GitHub} repository.\footnote{https://github.com/Jalilahmad4/Quantum-Opt.-for-Parameter-Estimation-in-Dynamical-Systems.git}

\paragraph{\textbf{Common Numerical and Optimization Settings}}
In all experiments, synthetic observations are generated by numerically solving
the corresponding dynamical system with prescribed true parameter values at a
finite collection of observation times. The nudged systems are initialized
from prescribed initial states that differ from the reference initial
conditions. During integration of the nudged systems, the discrete observations
are linearly interpolated to provide observation values at the internal time
points selected by the ODE solver. The model-specific reference
initial conditions, observation intervals, observed components, nudging
coefficients, and parameter search domains are stated in the corresponding
subsections.

For each unknown parameter, the cost functional is evaluated
at \(M_c=8\) coarse-grid points. The resulting coarse-grid cost functional
values are shifted and normalized to \([0,1]\) and used to fit a continuous
quadratic surrogate by weighted least squares with weighting parameter
\(\lambda=8\). The continuous surrogate is then evaluated on a refined grid
containing \(M_f=32\) points per unknown parameter without requiring additional
data-assimilation solves. Since \(32=2^5\), each unknown parameter is
represented using five binary variables. The coefficients of the corresponding
refined-grid QUBO representation are recovered by weighted least squares using
the same weighting parameter \(\lambda=8\). The numerical QUBO representation
error \(\mathrm{MSE}_{\mathrm{QUBO}}\) is monitored to verify that the binary
quadratic model reproduces the refined continuous surrogate up to numerical
precision.

The resulting QUBO problems are solved using QAOA on an IBM quantum simulator,
QAOA on the IBM Kingston QPU, and SQA implemented using D-Wave's
\texttt{PathIntegralAnnealingSampler}. The same QAOA implementation and
execution settings are used across all four model problems. For the SQA
experiments, \(5000\) reads and \(2000\) schedule points are used with random
seed \(123\), and the lowest-energy sampled bitstring is selected as the final
solution. For each optimization procedure, the returned low-energy bitstring
is decoded on the refined parameter grid to obtain the corresponding parameter
estimate.

For the two-parameter examples (SIS, SIR, and Lorenz--96), the coarse grid
contains \(8^2=64\) data-assimilation solves and the refined grid contains
\(32^2=1024\) binary-encoded parameter points. Since each parameter is
represented by five binary variables, the refined search uses
\(2\times5=10\) binary variables. For the Lorenz--63 model, the coarse grid
contains \(8^3=512\) data-assimilation solves and the refined grid contains
\(32^3=32768\) binary-encoded parameter points, corresponding to
\(3\times5=15\) binary variables. The common grid sizes, observations, and
binary encodings are summarized in Table~\ref{tab:experiment_summary}.

\begin{table}[h]
\centering
\begin{tabular}{l c c c c}
\hline
Model & Unknowns & Observed data & DA solves & Binary variables \\
\hline
SIS
& \(\beta,\gamma\)
& \(I(t)\)
& \(8^2=64\)
& \(10\) \\
SIR
& \(\beta,\gamma\)
& \(I(t)\)
& \(8^2=64\)
& \(10\) \\
Lorenz--63
& \(\sigma,\rho,\beta\)
& \(x(t)\)
& \(8^3=512\)
& \(15\) \\
Lorenz--96
& \(d_{u,5},d_{u,15}\)
& \(u_5(t),u_{15}(t)\)
& \(8^2=64\)
& \(10\) \\
\hline
\end{tabular}
\caption{Summary of the numerical experiments. In all cases, cost functional values are computed on a coarse parameter
grid, while the refined binary parameter search is represented as a QUBO and
solved using the selected optimization procedures.}
\label{tab:experiment_summary}
\end{table}

\subsection{SIS Model}
\label{SIS Model}

We first consider the susceptible--infected--susceptible (SIS) model \cite{kermack}. The population is divided into susceptible and infected
classes, represented by the population fractions \(S(t)\) and \(I(t)\),
respectively. Thus, $S(t)+I(t)=1.$ The dynamics are
\begin{equation}
\begin{cases}
\dot S = -\beta SI+\gamma I,\\
\dot I = \beta SI-\gamma I,
\end{cases}
\label{eq:sis_model}
\end{equation}
where \(\beta>0\) is the transmission rate and \(\gamma>0\) is the recovery rate.

Synthetic observations are generated using $S(0)=0.99875, I(0)=0.00125,$
with true parameters $\beta_{\mathrm{true}}=0.65,
\gamma_{\mathrm{true}}=0.25.$
The system is solved on \([0,40]\) using \(200\) observation times, and only the infected prevalence \(I(t)\) is treated as observed data.

For each candidate parameter pair, we solve the nudged SIS system
\begin{equation}
\begin{cases}
\dot{\tilde S}
=
-\beta\tilde S\tilde I+\gamma\tilde I,\\
\dot{\tilde I}
=
\beta\tilde S\tilde I-\gamma\tilde I
+\mu_I\bigl(I_{\mathrm{data}}(t)-\tilde I\bigr),
\end{cases}
\label{eq:sis_nudged}
\end{equation}
with \(\mu_I=0.5\). The cost functional is
\begin{equation}
C(\beta,\gamma)
=
\sum_{i=1}^{N_t}
\left|I_{\mathrm{data}}(t_i)-\tilde I(t_i;\beta,\gamma)\right|^2.
\label{eq:sis_cost}
\end{equation}
The parameter search domain is $\Theta_{\mathrm{s}}
=
[0.55,0.75]\times[0.18,0.32],$
corresponding to the parameters \((\beta,\gamma)\).

Table~\ref{tab:sis_params} reports the SIS parameter estimates obtained from
the IBM simulator, IBM Kingston QPU, and SQA. All three
implementations recover parameter values close to the true generating
parameters. The estimated value of \(\beta\) has relative error approximately
\(1.48\%\), while the estimated value of \(\gamma\) has relative error
approximately \(2.70\%\).

\begin{table}[h]
\centering
\small
\begin{tabularx}{0.95\textwidth}{c c X X X}
\hline
Parameter 
& True value 
& IBM Simulator 
& IBM QPU 
& SQA \\
\hline
\(\beta\)
& \(0.65\)
& \(0.659677\; (1.48\%)\)
& \(0.659677\; (1.48\%)\)
& \(0.653226\; (0.49\%)\) \\
\(\gamma\)
& \(0.25\)
& \(0.256774\; (2.70\%)\)
& \(0.256774\; (2.70\%)\)
& \(0.252258\; (0.90\%)\) \\
\hline
\end{tabularx}
\caption{Estimated SIS parameters obtained from the IBM simulator, IBM Kingston QPU, and SQA. Relative percentage errors are shown in parentheses.}
\label{tab:sis_params}
\end{table}

Figure~\ref{fig:sis_trajectories} compares the susceptible and infected population-fraction trajectories generated using the estimated parameters with the true SIS
trajectories. The estimated trajectories closely overlap with the true
solutions for both compartments over the full time interval. This shows that
the recovered parameters reproduce the SIS dynamics accurately, even though the
final parameter search is performed through the QUBO/Ising optimization
formulation.

\begin{figure}[h]
\centering
\includegraphics[width=0.92\textwidth]{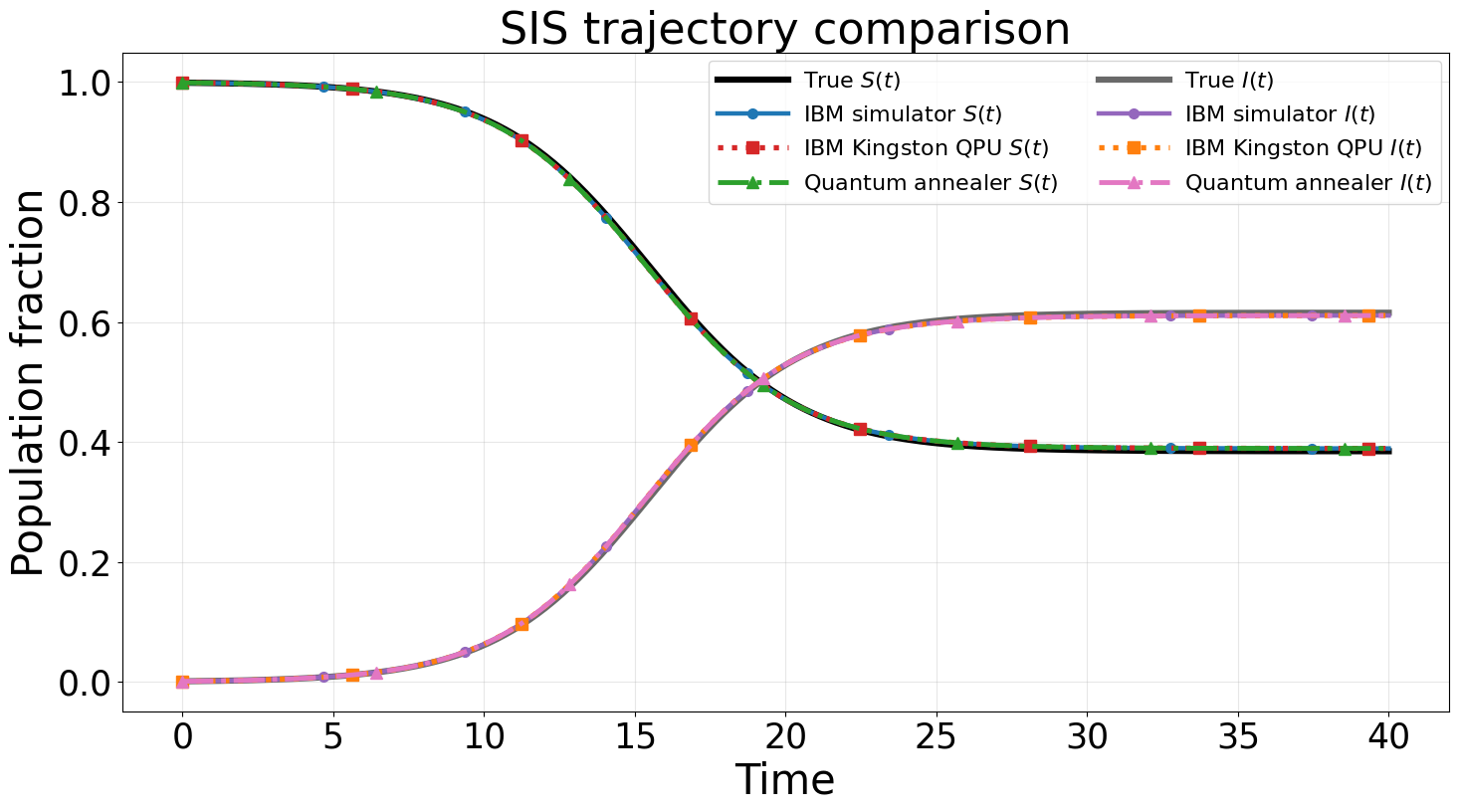}
\caption{SIS trajectory comparison for the susceptible and infected population fractions using the true parameters and the parameters estimated by the IBM simulator, IBM Kingston QPU, and SQA.}
\label{fig:sis_trajectories}
\end{figure}

\subsection{SIR Model}
\label{SIR Model}

We consider the classical susceptible--infected--recovered (SIR) model, which is commonly used to describe the spread of an infectious disease in a closed population \cite{kermack}. The population is divided into three compartments: susceptible \(S(t)\), infected \(I(t)\), and recovered \(R(t)\). The dynamics are governed by
\begin{equation}
\begin{cases}
\dot{S} = -\beta S I, \\
\dot{I} = \beta S I - \gamma I, \\
\dot{R} = \gamma I,
\end{cases}
\label{eq:sir_model}
\end{equation}
where \(\beta>0\) is the transmission rate and \(\gamma>0\) is the recovery rate. Throughout this subsection, the state variables are normalized population fractions, so that \(S(t), I(t), R(t)\in[0,1]\), and the total population satisfies \(S(t)+I(t)+R(t)=1\).

Synthetic infected-prevalence observations are generated by numerically solving Eq.~\eqref{eq:sir_model} with $\beta_{\mathrm{true}} = 0.32,
\gamma_{\mathrm{true}} = 0.12,$
and initial condition $S(0)=0.99,
I(0)=0.01,
R(0)=0.$
The system is solved over the time interval \([0,30]\). The infected component of the trajectory, denoted by \(I_{\mathrm{data}}(t)\), is treated as the observed prevalence data. Thus, the estimation procedure uses only observations of the infected population fraction, while \(\beta\) and \(\gamma\) are treated as unknown.

The parameter search domain is $\Theta_{\mathrm{s}}
=
[0.10,0.50]\times[0.05,0.25],$
corresponding to the parameters \((\beta,\gamma)\). For each coarse-grid parameter pair, the nudged SIR system is solved with nudging gain \(\mu=0.5\). Since the infected compartment is observed, the nudging term is applied to the infected equation:
\begin{equation}
\begin{cases}
\dot{\tilde{S}} = -\beta \tilde{S}\tilde{I}, \\
\dot{\tilde{I}} = \beta \tilde{S}\tilde{I} - \gamma \tilde{I}
+ \mu\big(I_{\mathrm{data}}(t)-\tilde{I}\big), \\
\dot{\tilde{R}} = \gamma \tilde{I}.
\end{cases}
\label{eq:sir_nudged}
\end{equation}
The data-misfit cost function is evaluated using only the observed infected-prevalence component:
\begin{equation}
C(\beta,\gamma)
=
\sum_{i=1}^{N_t}
\left|I_{\mathrm{data}}(t_i)-\tilde{I}(t_i;\beta,\gamma)\right|^2,
\label{eq:sir_cost}
\end{equation}
where \(N_t=200\) is the number of observation times.

Table~\ref{tab:sir_params} reports the SIR parameter estimates obtained from
the IBM simulator, IBM Kingston QPU, and SQA. All three
implementations recover the same parameter values, which are close to the true
generating parameters. The estimated transmission rate \(\beta\) has relative
error approximately \(0.20\%\), while the estimated recovery rate \(\gamma\)
has relative error approximately \(0.81\%\).

\begin{table}[h]
\centering
\small
\begin{tabularx}{0.95\textwidth}{c c X X X}
\hline
Parameter 
& True value 
& IBM Simulator 
& IBM QPU 
& SQA \\
\hline
\(\beta\)
& \(0.32\)
& \(0.319355\; (0.20\%)\)
& \(0.319355\; (0.20\%)\)
& \(0.319355\; (0.20\%)\) \\
\(\gamma\)
& \(0.12\)
& \(0.120968\; (0.81\%)\)
& \(0.120968\; (0.81\%)\)
& \(0.120968\; (0.81\%)\) \\
\hline
\end{tabularx}
\caption{Estimated SIR parameters obtained from the IBM simulator, IBM Kingston QPU, and SQA. Relative percentage errors are shown in parentheses.}
\label{tab:sir_params}
\end{table}
Figure~\ref{fig:sir_trajectories} compares the susceptible, infected, and
recovered population-fraction trajectories generated using the estimated parameters with the true
SIR trajectories. The estimated trajectories closely overlap with the true
solutions for all three compartments over the full time interval. This shows
that the recovered parameters reproduce the SIR dynamics accurately from
infected-prevalence observations, even though the final parameter search is
performed through the QUBO/Ising optimization formulation.

\begin{figure}[h]
\centering
\includegraphics[width=0.95\textwidth]{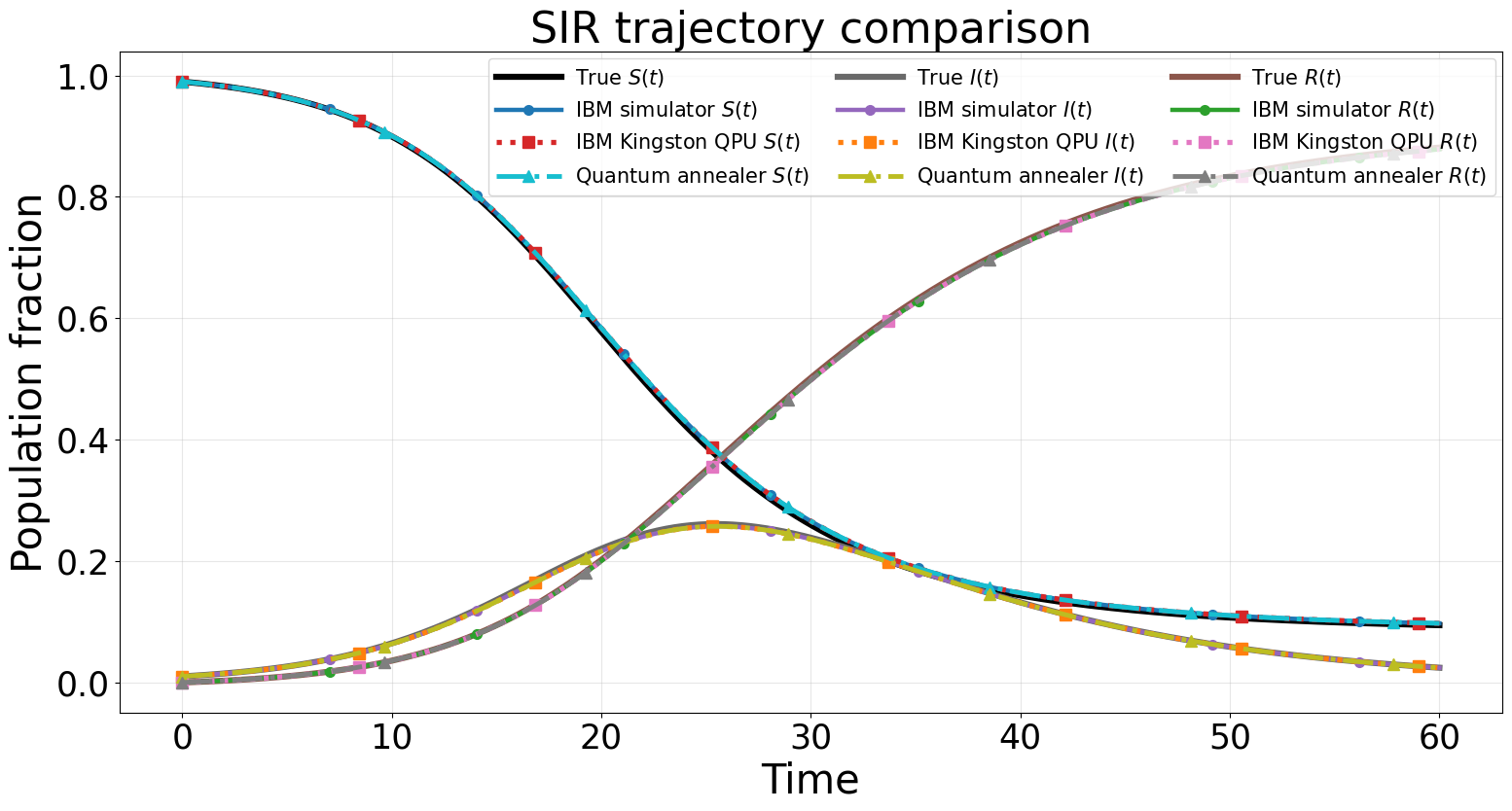}
\caption{SIR trajectory comparison for the susceptible, infected, and recovered population fractions using the true parameters and the parameters estimated by the IBM simulator, IBM Kingston QPU, and SQA.}
\label{fig:sir_trajectories}
\end{figure}

\subsection{Lorenz--63 System}
\label{lorenz63}

The Lorenz--63 system was introduced by Edward Lorenz in 1963 as a simplified mathematical model for atmospheric convection \cite{lorenz1963}. It is a three-dimensional nonlinear system of ordinary differential equations given by
\begin{equation}
\begin{cases}
    \dot{x}= \sigma (y - x), \\
    \dot{y}= x(\rho - z) - y, \\
    \dot{z}= xy - \beta z,
\end{cases}
\label{eq:lorenz_model}
\end{equation}
where \(x(t)\), \(y(t)\), and \(z(t)\) denote the state variables, and \(\sigma\), \(\rho\), and \(\beta\) are model parameters. In the physical interpretation of the model, \(x\) is proportional to the intensity of convective motion, \(y\) represents the horizontal temperature variation, and \(z\) represents the vertical temperature variation \cite{sparrow1982}. The parameter \(\sigma\) is the Prandtl number, \(\rho\) is related to the Rayleigh number, and \(\beta\) depends on the geometry of the convective layer.

For the Lorenz--63 experiment, we choose the true parameter values $\sigma_{\mathrm{true}} = 10,
\rho_{\mathrm{true}} = 30,
\beta_{\mathrm{true}} = \frac{8}{3} \approx 2.666667,$
for which the system exhibits chaotic dynamics. We use the initial condition $x(0)=1,
y(0)=1,
z(0)=1.$
Synthetic observations are generated by numerically solving Eq.~\eqref{eq:lorenz_model} over \([0,20]\) using \(400\) uniformly spaced time points. Only the \(x\)-component of the resulting trajectory, denoted by \(x_{\mathrm{data}}(t)\), is treated as observed data. Therefore, although the full Lorenz--63 system is used to generate the synthetic trajectory, the estimation procedure has access only to partial observations. The parameters \(\sigma\), \(\rho\), and \(\beta\) are all treated as unknown.

The parameter search domain is $\Theta_{\mathrm{s}}
=
[8,12]\times[28,32]\times[2.4,2.9],$
corresponding to the parameters \((\sigma,\rho,\beta)\). For each coarse-grid parameter triple, the nudged Lorenz--63 system is solved with nudging gain \(\mu=30\). Since only the \(x\)-component is observed, the nudging term is applied only to the \(x\)-equation:
\begin{equation}
\begin{cases}
    \dot{\tilde{x}}= \sigma(\tilde{y}-\tilde{x})
    + \mu\big(x_{\mathrm{data}}(t)-\tilde{x}\big), \\
    \dot{\tilde{y}}= \tilde{x}(\rho-\tilde{z})-\tilde{y}, \\
    \dot{\tilde{z}}= \tilde{x}\tilde{y}-\beta\tilde{z}.
\end{cases}
\label{eq:lorenz_nudged}
\end{equation}
The data-misfit cost functional is evaluated using only the observed \(x\)-component:
\begin{equation}
    C(\sigma,\rho,\beta)
    =
    \sum_{i=1}^{N_t}
    \left|x_{\mathrm{data}}(t_i)-\tilde{x}(t_i;\sigma,\rho,\beta)\right|^2,
\label{eq:lorenz_cost}
\end{equation}
where \(N_t=400\). The continuous quadratic surrogate fitted from the \(8^3\) coarse-grid cost functional values is evaluated on a refined \(32^3\) grid and encoded as a QUBO cost functional.

Table~\ref{tab:lorenz_params} reports the Lorenz--63 parameter estimates
obtained from the IBM simulator, IBM Kingston QPU, and SQA. The estimates are close to the true parameter values, with relative
errors ranging from about \(2.80\%\) to \(4.56\%\). These discrepancies are
expected in this chaotic and partially observed setting, where only the
\(x\)-component is used in the cost functional.

\begin{table}[H]
\centering
\small
\begin{tabularx}{0.95\textwidth}{c c X X X}
\hline
Parameter 
& True value 
& IBM simulator 
& IBM QPU 
& SQA \\
\hline
\(\sigma\)
& \(10.0\)
& \(9.67741\; (3.23\%)\)
& \(9.67741\; (3.23\%)\)
& \(9.67741\; (3.23\%)\) \\
\(\rho\)
& \(30.0\)
& \(31.22580\; (4.09\%)\)
& \(30.70967\; (2.80\%)\)
& \(31.09677\; (2.36\%)\) \\
\(\beta\)
& \(2.67\)
& \(2.54516\; (4.56\%)\)
& \(2.57741\; (3.35\%)\)
& \(2.593548\; (2.74\%)\) \\
\hline
\end{tabularx}
\caption{Estimated Lorenz--63 parameters obtained from the IBM simulator, IBM Kingston QPU, and SQA using only \(x\)-component observations. Relative percentage errors are shown in parentheses.}
\label{tab:lorenz_params}
\end{table}

Figure~\ref{fig:lorenz63_all} compares the phase-space trajectories generated using the estimated parameters with the true Lorenz--63 trajectory. The estimated trajectories reproduce the overall structure of the Lorenz attractor and remain close to the true attractor despite small differences in the recovered parameter values. Because the Lorenz--63 system is chaotic, pointwise agreement over long time intervals is not expected; instead, the phase-space comparison demonstrates that the estimated parameters preserve the characteristic geometry of the underlying chaotic dynamics.

\begin{figure}[h]
\centering
\includegraphics[width=0.5\textwidth]{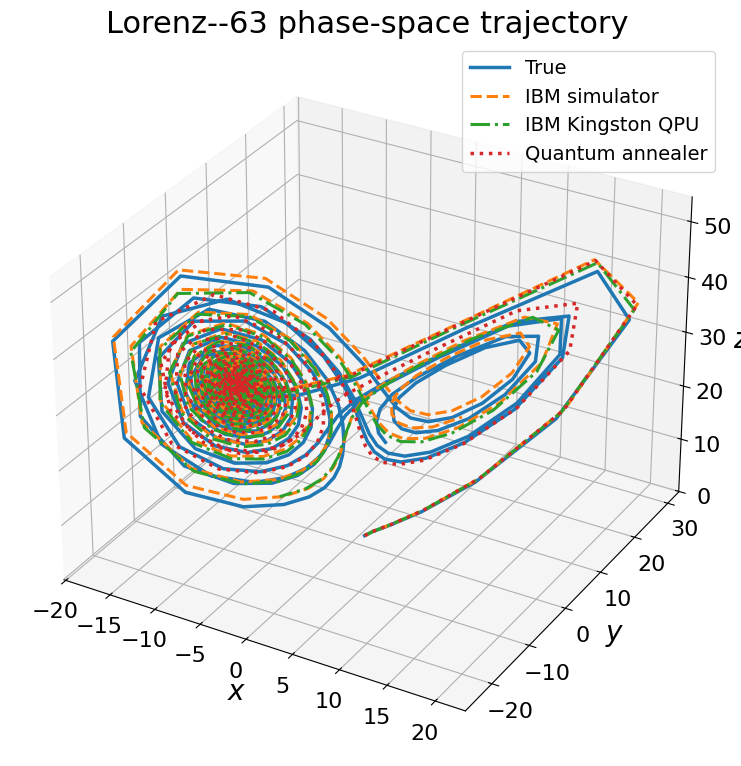}
\caption{Lorenz--63 phase-space trajectory comparison using the true parameters and the parameters estimated by the IBM simulator, IBM Kingston QPU, and SQA.}
\label{fig:lorenz63_all}
\end{figure}

\subsection{Lorenz--96 System}
\label{lorenz96}

We finally consider a high-dimensional, multiscale two-layer Lorenz--96 system. The original Lorenz--96
model is a standard simplified model for atmospheric dynamics and predictability
\cite{lorenz1996}. In the one-layer setting, the dynamics are
\begin{equation}
\frac{d u_k}{dt}
=
u_{k-1}(u_{k+1}-u_{k-2}) - d_k u_k + F,
\label{eq:l96_one_layer}
\end{equation}
where the indices are interpreted periodically, \(u_k(t)\) denotes the
large-scale state variable, \(d_k\) is a damping coefficient, and \(F\) is a
constant forcing term.

Following the two-layer Lorenz--96 formulation in \cite{Martinez2024}, we use
a model with \(K\) slow variables \(u_k(t)\), \(k=0,\ldots,K-1\), and \(J\)
fast variables \(v_{k,j}(t)\), \(j=0,\ldots,J-1\), associated with each slow
variable. The two-layer system is
\begin{equation}
\begin{cases}
\displaystyle
\dot{u}_k
=
u_{k-1}\big(u_{k+1}-u_{k-2}\big)
+
\sum_{j=0}^{J-1}\gamma_j v_{k,j}u_k
-
d_{u,k}u_k
+
F, \\[2mm]
\displaystyle
\dot{v}_{k,j}
=
-d_{v,k,j}v_{k,j}
-
\gamma_j u_k^2,
\end{cases}
\label{eq:l96_model}
\end{equation}
where the indices in the slow variables are interpreted periodically. Here
\(\gamma_j\) controls the coupling between the slow and fast variables, while
\(d_{u,k}\) and \(d_{v,k,j}\) are damping parameters for the slow and fast
variables, respectively.

In this experiment, we use $K=40, J=5, F=5,$
and the coupling vector $\gamma=(0.10,0.15,0.20,0.25,0.35).$
The true slow damping parameters are generated by
\begin{equation}
d_{u,k}^{\mathrm{true}}
=
1+0.7\cos\left(\frac{2\pi(k+1)}{5}\right),
\qquad k=0,\ldots,K-1,
\label{eq:l96_du_true}
\end{equation}
and the true fast damping parameters are
\[
d_{v,k,:}^{\mathrm{true}}=(0.2,0.5,1.0,2.0,5.0),
\qquad k=0,\ldots,K-1.
\]
The goal is to estimate the two slow damping parameters $d_{u,5} \text{ and }
d_{u,15}.$
The corresponding true values are
\[
d_{u,5}^{\mathrm{true}}=1.216312,
\qquad
d_{u,15}^{\mathrm{true}}=1.216312.
\]
All other parameters are fixed at their true values. The corresponding slow components \(u_5(t)\) and
\(u_{15}(t)\) are observed.

Synthetic observations are generated by solving Eq.~\eqref{eq:l96_model} on
\([0,40]\) using \(600\) uniformly spaced time points. The deterministic initial
condition is
\[
u_k(0)
=
F+0.2\sin\left(\frac{2\pi k}{K}\right),
\qquad
v_{k,j}(0)=0.1 .
\]
For each candidate parameter pair $\theta=(d_{u,5},d_{u,15}),$
we solve a nudged two-layer Lorenz--96 system. The nudging terms are applied
only to the observed slow equations \(u_5\) and \(u_{15}\). Thus,
\begin{equation}
\begin{cases}
\displaystyle
\dot{\tilde{u}}_k
=
\tilde{u}_{k-1}\big(\tilde{u}_{k+1}-\tilde{u}_{k-2}\big)
+
\sum_{j=0}^{J-1}\gamma_j \tilde{v}_{k,j}\tilde{u}_k
-
d_{u,k}\tilde{u}_k
+
F,
\qquad k\notin\{5,15\}, \\[2mm]
\displaystyle
\dot{\tilde{u}}_5
=
\tilde{u}_{4}\big(\tilde{u}_{6}-\tilde{u}_{3}\big)
+
\sum_{j=0}^{J-1}\gamma_j \tilde{v}_{5,j}\tilde{u}_5
-
d_{u,5}\tilde{u}_5
+
F
+
\mu\big(u_{5,\mathrm{data}}(t)-\tilde{u}_5\big), \\[2mm]
\displaystyle
\dot{\tilde{u}}_{15}
=
\tilde{u}_{14}\big(\tilde{u}_{16}-\tilde{u}_{13}\big)
+
\sum_{j=0}^{J-1}\gamma_j \tilde{v}_{15,j}\tilde{u}_{15}
-
d_{u,15}\tilde{u}_{15}
+
F
+
\mu\big(u_{15,\mathrm{data}}(t)-\tilde{u}_{15}\big), \\[2mm]
\displaystyle
\dot{\tilde{v}}_{k,j}
=
-d_{v,k,j}\tilde{v}_{k,j}
-
\gamma_j \tilde{u}_k^2 .
\end{cases}
\label{eq:l96_nudged}
\end{equation}
In the numerical experiment, the nudging gain is \(\mu=15\).

The cost functional is computed using the two observed slow components:
\begin{equation}
C(d_{u,5},d_{u,15})
=
\sum_{i=1}^{N_t}
\left|u_{5,\mathrm{data}}(t_i)-\tilde{u}_5(t_i)\right|^2
+
\sum_{i=1}^{N_t}
\left|u_{15,\mathrm{data}}(t_i)-\tilde{u}_{15}(t_i)\right|^2 .
\label{eq:l96_cost}
\end{equation}
Here \(N_t=600\). The parameter search interval for each unknown parameter is chosen as $\theta_j\in
\left[
\theta_j^{\mathrm{true}}-0.5,\,
\theta_j^{\mathrm{true}}+0.5
\right],$
with a positive lower bound imposed when necessary. The coarse
data-assimilation grid contains \(8^2=64\) expensive model evaluations, and the
refined QUBO search uses a \(32\times32\) grid. Since each parameter is encoded
using five bits, the final binary optimization problem uses \(10\) binary
variables.

Table~\ref{tab:l96_params} reports the recovered parameter values obtained from
the IBM simulator, IBM Kingston QPU, and SQA. All three
implementations recover the same estimates for \(d_{u,5}\) and \(d_{u,15}\),
with relative errors approximately \(1.33\%\).

\begin{table}[H]
\centering
\small
\begin{tabularx}{0.95\textwidth}{c c X X X}
\hline
Parameter 
& True value 
& IBM simulator 
& IBM QPU 
& SQA \\
\hline
\(d_{u,5}\)
& \(1.216312\)
& \(1.200183\; (1.33\%)\)
& \(1.200183\; (1.33\%)\)
& \(1.200183\; (1.33\%)\) \\
\(d_{u,15}\)
& \(1.216312\)
& \(1.200183\; (1.33\%)\)
& \(1.200183\; (1.33\%)\)
& \(1.200183\; (1.33\%)\) \\
\hline
\end{tabularx}
\caption{Estimated two-layer Lorenz--96 parameters obtained from the IBM simulator, IBM Kingston QPU, and SQA. Relative percentage errors are shown in parentheses.}
\label{tab:l96_params}
\end{table}


\subsection{Surrogate and QUBO Representation Errors}
\label{subsec:fitting_errors}

Table~\ref{tab:fitting_errors} summarizes the continuous-surrogate fitting
errors and numerical QUBO representation errors for the four model problems.
The quantities \(\mathrm{RMSE}_{\mathrm{LS}}\) and
\(e_{\infty,\mathrm{LS}}\) measure the accuracy of the continuous quadratic
surrogate in approximating the normalized coarse-grid cost functional values,
while \(\mathrm{MSE}_{\mathrm{QUBO}}\) measures the numerical accuracy of the
QUBO representation on the refined grid.

\begin{table}[h]
\centering
\begin{tabular}{l c c c}
\hline
Model
& \(\mathrm{RMSE}_{\mathrm{LS}}\)
& \(e_{\infty,\mathrm{LS}}\)
& \(\mathrm{MSE}_{\mathrm{QUBO}}\) \\
\hline
SIS
& \(4.258770\times10^{-3}\)
& \(9.251285\times10^{-2}\)
& \(1.711676\times10^{-30}\) \\
SIR
& \(2.034177\times10^{-2}\)
& \(1.216511\times10^{-1}\)
& \(1.379741\times10^{-31}\) \\
Lorenz--63
& \(9.213976\times10^{-2}\)
& \(8.081225\times10^{-1}\)
& \(1.367646\times10^{-28}\) \\
Lorenz--96 & \(2.371029\times10^{-3}\) & \(2.344146\times10^{-2}\) & \(7.624507\times10^{-31}\) \\
\hline
\end{tabular}
\caption{Continuous-surrogate fitting errors and numerical QUBO
representation errors for the four numerical experiments.}
\label{tab:fitting_errors}
\end{table}

The QUBO representation errors are near numerical precision for the reported
models, indicating that the QUBO coefficient-recovery step introduces no
meaningful additional approximation. 
 

\section{Conclusions}
\label{D}

This work developed a hybrid classical--quantum framework for parameter estimation in nonlinear dynamical systems. The proposed approach combines data-assimilation-augmented cost functional construction with a QUBO optimization stage. All dynamical simulations are performed classically: the nudged ODE system is solved only on a prescribed coarse parameter grid, and the resulting data-misfit cost functional values are used to construct a quadratic surrogate. This surrogate is then evaluated on a refined grid, encoded as a QUBO, mapped to an Ising Hamiltonian, and approximately minimized using quantum optimizers. In this way, the quantum component is used only for the discrete optimization stage, while the potentially difficult task of simulating nonlinear ODEs remains entirely classical. This separation makes the framework compatible with quantum optimization methods and avoids the need for quantum state tomography or direct quantum simulation of nonlinear dynamics.

The main contribution of the paper is a coarse-to-refined quantum-assisted parameter estimation strategy that reduces the number of expensive data-assimilation solves while still allowing a refined binary search over the parameter domain. The method provides a systematic connection between nudging-based inverse problems, quadratic surrogate modeling, QUBO formulations, Ising Hamiltonians, and quantum optimizers. Theoretical results were also presented to clarify the effect of surrogate approximation error and refined-grid discretization on the recovered parameter estimate. These results show that, when the surrogate uniformly approximates the refined-grid target cost functional and the grid resolution is sufficiently fine, the minimizer of the QUBO surrogate provides a near-optimal estimate for the underlying discrete parameter estimation problem.

Numerical experiments demonstrated the performance of the proposed framework on four representative systems: the SIS and SIR epidemic models, the Lorenz--63 system, and the two-layer Lorenz--96 system. For the SIS model, the recovered parameters had relative errors of about
\(1.48\%\) for \(\beta\) and \(2.70\%\) for \(\gamma\). For the SIR model, the relative errors were below \(1\%\) for both parameters. In the Lorenz--63 example, where only the \(x(t)\) component was observed, the parameter errors were larger, as expected for a chaotic and partially observed system, but the recovered trajectories still reproduced the overall Lorenz attractor structure. In the two-layer Lorenz--96 example, the two damping parameters \(d_{u,5}\) and \(d_{u,15}\) were recovered with relative errors of about \(1.33\%\), using observations only from \(u_5(t)\) and \(u_{15}(t)\).

An important feature of these experiments is that the estimation procedure uses only partial and time-discrete observations. The SIS and SIR examples use infected-prevalence observations rather than full-state observations. The Lorenz--63 experiment uses only the \(x(t)\) component to estimate all three model parameters, and the Lorenz--96 experiment uses only two observed slow variables, \(u_5(t)\) and \(u_{15}(t)\), to estimate the corresponding damping parameters. Moreover, the observations are available only at finitely many time points, not continuously in time. Thus, the results demonstrate that the proposed framework can be used in realistic observation settings where only selected components of the state are measured at discrete sampling times. Additional numerical tests, not reported here, indicated that increasing the
number of observation time points generally improved the accuracy of the recovered
parameter values. This suggests that the amount and temporal resolution of the
available data can influence the accuracy of the proposed parameter estimation
framework.

The IBM Kingston experiments show that the approach can be executed on real
quantum hardware, although the recovered estimates may be affected by
sampling noise and hardware noise. The SQA results provide an additional
comparison for the same QUBO energy landscape. The computational time of the final optimization stage depends on the selected solver and its execution settings. The classical setup stage, including data-assimilation solves, surrogate construction, QUBO fitting, circuit
construction, and transpilation, remains part of the total computational cost.
Therefore, runtime comparisons should be interpreted in the context of the
full hybrid workflow.

Another important point is that the present proof-of-concept experiments use a relatively small number of qubits. The two-parameter examples (SIS, SIR, Lorenz--96) use \(10\) binary variables, corresponding to five bits per parameter, while the Lorenz--63 example uses \(15\) binary variables for three unknown parameters. Even with this modest number of qubits, the refined search corresponds to \(32^2=1024\) candidate parameter values for two-parameter problems and \(32^3=32768\) candidate parameter values for the three-parameter Lorenz--63 problem. This illustrates the potential benefit of the binary encoding: increasing the number of available qubits would allow finer parameter grids, larger parameter spaces, or both. Thus, as quantum hardware improves, the same framework could support higher-resolution searches and more complex inverse problems.

Future work will focus on improving the scalability, robustness, and practical applicability of the framework. In particular, we plan to develop adaptive coarse-grid sampling and local refinement strategies so that data-assimilation solves are concentrated in promising regions with small cost functional values, especially when the cost functional landscape is nonquadratic, multimodal, or difficult to resolve on a uniform grid. We also plan to explore more expressive surrogate models, such as higher-order, piecewise, sparse, or physics-informed surrogates, while preserving a QUBO-compatible formulation. Another important direction is to study scalability with respect to the number
of unknown parameters and the number of bits used per parameter, since larger quantum devices would allow finer parameter grids and higher-dimensional inverse problems.  It will also be important to investigate identifiability more systematically for partially observed dynamical systems, including the chaotic Lorenz--63 system and the high-dimensional two-layer Lorenz--96 system, and to assess the effect of noisy observations, model error, correlated parameters, hardware noise, queue time, circuit depth, sampling variability, and shot count. More broadly, applying the framework to noisy real-world data, larger-scale dynamical systems, and PDE-constrained inverse problems would further clarify the practical role of quantum optimization in parameter estimation.
Additionally, while we employ QAOA and quantum annealing in this work because of their natural compatibility with QUBO formulations and their applicability to NISQ devices, future studies should investigate alternative quantum optimization approaches, including QHD, DQI, and other emerging methods, on near-term and fault-tolerant quantum platforms.

\section*{Acknowledgments}
The authors gratefully acknowledge support from the University of Maryland, Baltimore County (UMBC) Strategic Awards for Research Transitions (START) grant (PI: Animikh Biswas). Mohammadhossein Mohammadisiahroudi was also supported by UMBC Summer Research Faculty Fellowship (SURFF) grant for this project. The quantum computing experiments reported in this work used IBM Quantum services, including simulator and quantum hardware access. The simulated quantum annealing (SQA) experiments were performed using the D-Wave Ocean software and its \texttt{PathIntegralAnnealingSampler}.

\raggedbottom
\bibliographystyle{unsrt}
\bibliography{references}
\end{document}

%% file: ex_shared.tex
\usepackage{lipsum}
\usepackage{amsfonts}
\usepackage{graphicx}
\usepackage{epstopdf}
\usepackage{algorithm}
\usepackage{algpseudocode}
\usepackage{float}
\usepackage{supertabular}
\usepackage{mathtools}
\usepackage{amsmath}
\usepackage{supertabular}
\usepackage{multicol}
\usepackage{subcaption}
\usepackage{mathtools}
\ifpdf
  \DeclareGraphicsExtensions{.eps,.pdf,.png,.jpg}
\else
  \DeclareGraphicsExtensions{.eps}
\fi

\newsiamremark{remark}{Remark}
\newsiamremark{hypothesis}{Hypothesis}
\crefname{hypothesis}{Hypothesis}{Hypotheses}
\newsiamthm{claim}{Claim}
\newsiamremark{fact}{Fact}
\crefname{fact}{Fact}{Facts}

\headers{Quantum Optimization Framework for Parameter Estimation}{M. J. Ahmad, M. Mohammadisiahroudi, A. Biswas and K. Hoffman}

\title{A quantum optimization framework for data--assimilation--augmented parameter estimation\thanks{Submitted to the editors 08.12.2026.
}}

\author{
Muhammad Jalil Ahmad\thanks{
Department of Mathematics and Statistics, University of Maryland Baltimore County, USA
(\email{LS47576@umbc.edu}).
}
\and
Mohammadhossein Mohammadisiahroudi\footnotemark[2] \thanks{
Quantum Science Institute, University of Maryland Baltimore County, USA.
} 
\and Animikh Biswas\footnotemark[2]
\and Kathleen Hoffman\footnotemark[2]
}

\usepackage{amsopn}
